\documentclass[draftcls, 12pt,onecolumn]{IEEEtran}

\usepackage{cite}
\usepackage[cmex10]{amsmath}
\usepackage{amssymb}
\usepackage{amsthm}
\usepackage{mathrsfs}
\usepackage{bm}
\usepackage[mathscr]{eucal}
\usepackage{amssymb,amsmath,amsthm,amsfonts,latexsym}
\usepackage{amsmath,graphicx,bm,xcolor,url}
\usepackage{graphicx}
\usepackage{latexsym}
\usepackage{CJK}
\usepackage{indentfirst}
\usepackage{geometry}
\usepackage{psfrag}
\usepackage{setspace}
\usepackage{algorithmic}
\usepackage{algorithmic, cite}
\usepackage{algorithm}
\usepackage{array}
\usepackage{mdwmath}
\usepackage{mdwtab}
\usepackage{eqparbox}
\usepackage{url}
\usepackage{epstopdf}
\usepackage{epsfig,epsf,psfrag}
\usepackage{fixltx2e}
\usepackage{verbatim}
\usepackage{textcomp}
\usepackage{psfrag} 

\usepackage{caption}
\usepackage{subcaption}

\theoremstyle{plain}
\newtheorem{mythe}{Theorem}
\theoremstyle{remark}

\theoremstyle{plain}

\theoremstyle{remark}

\theoremstyle{plain}

\theoremstyle{remark}
\newtheorem{myrem}{Remark}
\theoremstyle{remark}

\theoremstyle{remark}

\theoremstyle{remark}

\theoremstyle{remark}

\theoremstyle{remark}

\catcode`~=11 \def\UrlSpecials{\do\~{\kern -.15em\lower .7ex\hbox{~}\kern .04em}} \catcode`~=13

\allowdisplaybreaks[4]

\newcommand{\calC}{\mathcal{C}}

\newcommand{\bb}{\mathbf{b}}
\newcommand{\bB}{\mathbf{B}}

\newcommand{\bi}{\mathbf{i}}
\newcommand{\bI}{\mathbf{I}}

\newcommand{\boldm}{\mathbf{m}}
\newcommand{\bM}{\mathbf{M}}

\newcommand{\bQ}{\mathbf{Q}}

\newcommand{\bR}{\mathbf{R}}

\newcommand{\bS}{\mathbf{S}}

\newcommand{\bT}{\mathbf{T}}
\newcommand{\bu}{\mathbf{u}}
\newcommand{\bU}{\mathbf{U}}

\newcommand{\bX}{\mathbf{X}}

\newcommand{\bbC}{\mathbb{C}}

\newcommand{\bbR}{\mathbb{R}}

\DeclareMathAlphabet{\mathbsf}{OT1}{cmss}{bx}{n}
\DeclareMathAlphabet{\mathssf}{OT1}{cmss}{m}{sl}

\DeclareSymbolFont{bsfletters}{OT1}{cmss}{bx}{n}
\DeclareSymbolFont{ssfletters}{OT1}{cmss}{m}{n}
\DeclareMathSymbol{\bsfGamma}{0}{bsfletters}{'000}
\DeclareMathSymbol{\ssfGamma}{0}{ssfletters}{'000}
\DeclareMathSymbol{\bsfDelta}{0}{bsfletters}{'001}
\DeclareMathSymbol{\ssfDelta}{0}{ssfletters}{'001}
\DeclareMathSymbol{\bsfTheta}{0}{bsfletters}{'002}
\DeclareMathSymbol{\ssfTheta}{0}{ssfletters}{'002}
\DeclareMathSymbol{\bsfLambda}{0}{bsfletters}{'003}
\DeclareMathSymbol{\ssfLambda}{0}{ssfletters}{'003}
\DeclareMathSymbol{\bsfXi}{0}{bsfletters}{'004}
\DeclareMathSymbol{\ssfXi}{0}{ssfletters}{'004}
\DeclareMathSymbol{\bsfPi}{0}{bsfletters}{'005}
\DeclareMathSymbol{\ssfPi}{0}{ssfletters}{'005}
\DeclareMathSymbol{\bsfSigma}{0}{bsfletters}{'006}
\DeclareMathSymbol{\ssfSigma}{0}{ssfletters}{'006}
\DeclareMathSymbol{\bsfUpsilon}{0}{bsfletters}{'007}
\DeclareMathSymbol{\ssfUpsilon}{0}{ssfletters}{'007}
\DeclareMathSymbol{\bsfPhi}{0}{bsfletters}{'010}
\DeclareMathSymbol{\ssfPhi}{0}{ssfletters}{'010}
\DeclareMathSymbol{\bsfPsi}{0}{bsfletters}{'011}
\DeclareMathSymbol{\ssfPsi}{0}{ssfletters}{'011}
\DeclareMathSymbol{\bsfOmega}{0}{bsfletters}{'012}
\DeclareMathSymbol{\ssfOmega}{0}{ssfletters}{'012}

\newcommand{\tili}{\widetilde{i}}

\newcommand{\tilv}{\widetilde{v}}

\newcommand{\bGamma}{\bm{\Gamma}}

\def\norm#1{\left\| #1 \right\|}
\def\norm2#1{\left\| #1 \right\|_2}
\def\norm22#1{\left\| #1 \right\|_2^2}

\DeclareMathOperator{\diag}{diag}

\DeclareMathOperator{\rank}{rank}

\newcommand{\qednew}{\nobreak \ifvmode \relax \else
      \ifdim\lastskip<1.5em \hskip-\lastskip
      \hskip1.5em plus0em minus0.5em \fi \nobreak
      \vrule height0.75em width0.5em depth0.25em\fi}

\allowdisplaybreaks[1]
\usepackage{float}
\usepackage{cite}

\usepackage[normalem]{ulem}

\DeclareMathOperator{\Tr}{Tr}

\newcommand{\sft}[1]{{\sf{} #1}}

\title{Magnetic Beamforming for Wireless Power Transfer}
\author{Gang~Yang, Mohammad~ R.~Vedady~Moghadam, and Rui~Zhang 
\thanks{Gang~Yang, Mohammad~ R.~Vedady~Moghadam, and Rui~Zhang are with the Department of Electrical and Computer Engineering, National University of Singapore (e-mail:\{eleyagg, elezhang\}@nus.edu.sg, vedady.m@u.nus.edu).}} 

\begin{document}

\maketitle

\begin{abstract}
Magnetic resonant coupling (MRC) is an efficient method for realizing the near-field wireless power transfer (WPT). The use of multiple transmitters (TXs) each with one coil can be applied to enhance the WPT performance by focusing the magnetic fields from all TX coils in a beam toward the receiver (RX) coil, termed as ``magnetic beamforming''. In this paper, we study the optimal magnetic beamforming for an MRC-WPT system with multiple TXs and a single RX. We formulate an optimization problem to jointly design the currents flowing through different TXs so as to minimize the total power drawn from their voltage sources, subject to the minimum power required by the RX load as well as the TXs' constraints on the peak voltage and current. For the special case of identical TX resistances and neglecting all TXs' constraints on the peak voltage and current, we show that the optimal current magnitude of each TX is proportional to the mutual inductance between its TX coil and the RX coil. In general, the problem is a non-convex quadratically constrained quadratic programming (QCQP) problem, which is reformulated as a semidefinite programming (SDP) problem. We show that its semidefinite relaxation (SDR) is tight. Numerical results show that magnetic beamforming significantly enhances the deliverable power as well as the WPT efficiency over the uncoordinated benchmark scheme of equal current allocation.
\end{abstract}

%

\section{Introduction}\label{introduction}
Near-filed wireless power transfer (WPT) has been attracting growing interest, due to its high efficiency for power transmission. Near-field WPT is realized by inductive coupling (IC) \cite{MurakamiSatoh96, KimCho01} for short-range applications in centimeters, or magnetic resonant coupling (MRC) \cite{KursSoljatic07, HuiLee14} for mid-range applications up to a couple of meters. With the short-range WPT technology reaching the stage of commercial use, the mid-range WPT technology has been gathering momentum in the last decade. The recent progress on mid-range WPT is referred to the review paper \cite{HuiLee14} and the references therein.

The MRC-WPT system with multiple transmitters (TXs) and/or multiple receivers (RXs) has been studied in the literature \cite{YoonLing11, AhnHong13, JadidianKatabi14, RezaZhang:C15}. The MRC-WPT system with only two TXs and one RX is studied in \cite{YoonLing11, AhnHong13}, while the obtained results cannot be readily extended to the case of more than two TXs. Recently, an ``Magnetic MIMO'' charging system~\cite{JadidianKatabi14} can charge a phone inside of a user's pocket 40cm away from the array of TX coils, independently of the phone's orientation. For an MRC-WPT system with multiple RXs, the load resistances of the RXs are jointly optimized in \cite{RezaZhang:C15} to minimize the total transmit power and address the ``near-far'' fairness problem. Deploying multiple TXs can help focus the magnetic fields on the RX~\cite{JadidianKatabi14}, in a manner analogous to beamforming in far-field wireless communications~\cite{GershmanSPM10}. However, to our best knowledge, there has been no prior work that designs the magnetic beamforming from a signal processing and optimization perspective, for an MRC-WPT system with arbitrary number of TXs, which thus motivates our work.

In this paper, as shown in Fig.~\ref{fig:Fig1}, we consider an MRC-WPT system with a single RX and multiple TXs of which their currents (or equivalently source voltages) can be adjusted such that the magnetic fields are constructively combined at the RX, thus achieving a magnetic beamforming gain. We formulate a problem to minimize the total power drawn from the voltage sources of all TXs by designing the currents flowing through different TXs, subject to the minimum power required by the RX load as well as the TXs' constraints on the peak voltage and current. For the special case of identical TX resistances and neglecting the TXs' constraints, the optimal current magnitude of each TX is shown to be proportional to the mutual inductance between its TX coil and the RX coil. In general, our formulated problem is a non-convex quadratically constrained quadratic programming (QCQP) problem. It is recast as a semidefinite programming (SDP) problem. We show that its semidefinite relaxation (SDR) \cite{LuoZhangSPM10, HuangPalomar10} is tight, i.e., the existing optimal solution to the SDR problem is always rank-one. The optimal solution to the QCQP problem can thus be obtained via standard convex optimization methods~\cite{ConvecOptBoyd04}. Numerical results show that magnetic beamforming significantly enhances the deliverable power and the WPT efficiency over the benchmark scheme of equal current allocation. 
\section{System Model} \label{system_model}
As shown in Fig.~\ref{fig:Fig1}, we consider an MRC-WPT system with $N \geq 1$ TXs each equipped with a single coil, indexed by $n$, and a single RX with a coil, index as $0$. Each TX $n$ is connected to a stable energy source supplying sinusoidal voltage over time given by $\tilv_{ n}(t) = {\mathrm{Re}} \{v_{ n} e^{jwt}\}$, with the complex $v_{ n}$ denoting an adjustable voltage and $w > 0$ denoting its operating angular frequency. Let $\tili_{ n}(t) = {\mathrm{Re}} \{i_{ n} e^{jwt}\}$ denote the steady-state current flowing through TX $n$, with the complex current $i_{ n}$. This current produces a time-varying magnetic flux in the $n$-th TX coil, which passes through the RX coil and induces time-varying currents in it. Let $\tili_{0}(t) = {\mathrm{Re}} \{i_{ 0} e^{jwt}\}$ be the steady-state current at the RX, with the complex current $i_{ 0}$.

We use $M_{n0}$ and $M_{n_1 n_2}$ to denote the mutual inductance between the $n$-th TX coil and the RX coil, and the mutual inductance between the $n_1$-th TX coil and the $n_2$-th TX coil, respectively. Each mutual inductance plays the role of magnetic channel between a pair of coils, and is a positive or negative real number depending on the physical characteristics, the relative distance, and the orientations of a pair of coils, etc. \cite{JadidianKatabi14}.


\begin{figure} [htbp]
\centering
\includegraphics[width=.65\columnwidth]{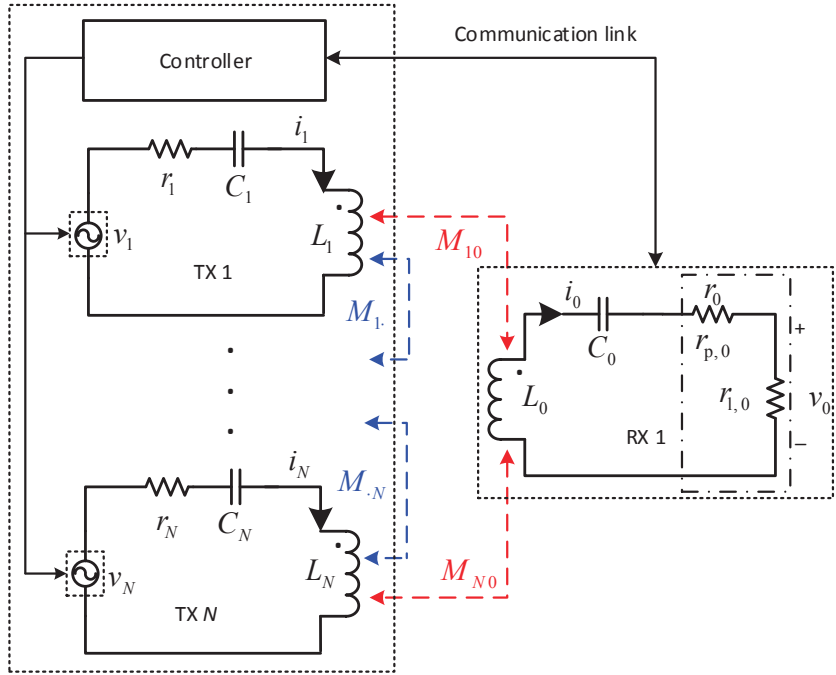}
\caption{System model of an MRC-WPT system}
\label{fig:Fig1}
\end{figure}



We denote the self-inductance and the capacity of the $n$-th TX coil (RX coil) by $L_{ n} >0$ ($L_0 > 0$) and $C_{ n} >0$ ($C_0 > 0$), respectively. The capacitors are chosen such that all TXs and the RX have the same resonant angular frequency $w$. We use $r_{ n} >0$ to denote the sum of the source resistance and the internal parasitic resistance of the $n$-th TX. The resistance of the RX, denoted by $r_{ 0} >0$, consists of the parasitic resistance $r_{ \sft p, 0} >0$ and the load resistance $r_{ \sft l, 0} >0$, i.e., $r_{ 0} = r_{ \sft p, 0} + r_{ \sft l, 0}$. The load is assumed to be purely resistive.




In this paper, we assume that all electrical parameters are fixed and known by the central controller. For analytical convenience, we treat the complex currents $i_{ n}$'s flowing through TXs as design parameters\footnote{In practice, it is more convenient to use a voltage source rather than a current source. However, by using standard circuit theory, one can easily compute the source voltages $v_{ n}$'s that generate the the required currents $i_{ n}$'s.}, which are adjusted to realize magnetic beamforming.

By applying Kirchhoff's circuit law to the RX, we obtain its current $i_{ 0}$ as follows
\begin{align}
 i_{ 0} &= \frac{jw}{r_0} \sum \limits_{n=1}^N M_{0n} i_{ n}. \label{eq:current_receiver_q}
\end{align}

\noindent Define the vector of mutual inductances between the RX coil and all TX coils as $\boldm = [M_{01} \; M_{02} \; \cdots \; $ $ M_{0N}]^H$.
From \eqref{eq:current_receiver_q}, the power delivered to the RX load is
\begin{align}
  p_0 &= \frac{1}{2} |i_0|^2 r_{ \sft l, 0} =\frac{w^2 r_{ \sft l, 0} }{2 r_ 0^2} {\bi}^H \boldm \boldm^H {\bi}.
  \label{eq:load_power_vec_complex}
\end{align}

By applying Kirchhoff's circuit law to each TX $n$, we obtain its source voltage as follows
\begin{align}
  v_n&=r_n i_n + j w \sum \limits_{k=1, \neq n}^N M_{nk} i_k - j w  M_{n0} i_0 \nonumber \\
  &= \left(\! r_n \!+\!   \frac{M_{n0}^2 w^2}{r_0} \!\right) i_n \!+\! \sum \limits_{k \neq n}   \left(\! j w M_{nk} \!+\!   \frac{M_{n0} M_{0k} w^2}{r_0} \! \right) i_k. \label{eq:voltage_transmitter_n}
\end{align}

\noindent We define the $N$-order square matrix $\bB $ as follows
\begin{align}
\bB  &= {\overline{\bB}}  + j {\widehat{\bB}},\label{eq_complex_B}
\end{align}
where the elements of the real-part matrix ${\overline{\bB}}$ and the imaginary-part matrix $\widehat{\bB}$ are given by
\begin{align}
 {\overline{B}}_{nk}  &=
  \left\{ \begin{array}{cl}
  r_{ n} +   \frac{M_{n0}^2 w^2}{r_0} , &\mbox{if}\; k=n\\
    \frac{M_{n0} M_{0k} w^2}{r_0}, \qquad &\mbox{otherwise} \\
  \end{array}
  \right. \label{eq_complex_B_real}\\
  {\widehat{B}}_{nk} &= \left\{ \begin{array}{cl}
  0, &\mbox{if}\; k=n\\
  - w M_{ nk}, \qquad &\mbox{otherwise} \\
  \end{array}
  \right.\label{eq_complex_B_imaginar}
 \end{align}
The matrices $\bB , \; {\overline{\bB}} $ and ${\widehat{\bB}}$ are symmetric, due to the fact that $M_{ nk} = M_{ kn}, \; \forall n \neq k$. Denote the $n$-th column of the matrices $\bB , \; \overline{\bB} , \; \widehat{\bB}$ by $\bb_n ,\; \overline{\bb}_n , \; \widehat{\bb}_n$, respectively. Moreover, the matrix ${\overline{\bB}} $ is positive semidefinite (PSD), as it can be rewritten as follows
 \begin{align}
   {\overline{\bB}} = \bR + \frac{w^2 \boldm \boldm^H}{r_0}. \label{eq:barB}
 \end{align}


\noindent The source voltage of each TX $n$ can be equivalently rewritten as
\begin{align}
  v_{ n}= {\bb}_n^H  \bi. \label{eq:voltage_transmitter_n_vector}
\end{align}

From \eqref{eq_complex_B} and \eqref{eq:voltage_transmitter_n_vector}, the total power drawn from the sources of all TXs, denoted by $p$, is derived as follows 
\begin{align}
  p &=\frac{1}{2} {\mathrm{Re}} \left\{ \sum \limits_{n=1}^N \bi^H {\bb}_n  i_{ n}\right\}
    = \frac{1}{2}  {{\bi}}^H  {\overline{\bB}}  {{\bi}}. \label{eq:p_transmit_general_vec_complex}
\end{align}

\begin{myrem}\label{remark:M_tx_p}
  From \eqref{eq_complex_B_real} and~\eqref{eq:p_transmit_general_vec_complex}, we observe that the total power drawn from all TXs' voltage sources depends on the mutual inductance $M_{n0}$ between the coils of each TX $n$ and the RX, but it is independent of the mutual inductance $M_{ nk}$ between any pair of TX coils.
\end{myrem}

\section{Problem Formulation}\label{sec: formulation}
In this section, we formulate a problem to minimize the total power drawn from the voltage sources of all TXs by jointly designing the currents $\bi$ flowing through TXs, subject to practical constraints of the MRC-WPT system. Particularly, we consider the following constraints: first, the power delivered to the RX load should exceed a given power level $\beta_0 >0$, i.e., $p_0 \geq \beta_0$; second, the maximally allowable amplitude of the source voltage $v_{ n}$ is $V_{ n}$. i.e., $|v_n| \leq V_n$; third, the maximally allowable amplitude of the current $i_{ n}$ is $A_{ n}$, i.e., $|i_n| \leq A_n$. Let $\bQ_n$ be the rank-one matrix with the $n$-th diagonal element as one and all other elements as zero. From \eqref{eq:load_power_vec_complex}, \eqref{eq:voltage_transmitter_n_vector} and \eqref{eq:p_transmit_general_vec_complex}, the problem is thus formulated as follows
\begin{subequations}\label{eq:optimP1}
\begin{align}
   \mathrm{(P1)}: \  \ \ &\underset{ {\bi} \in \calC^N }{\text{min}} \ \
    \frac{1}{2}   {{\bi}}^H  {\overline{\bB}} {{\bi}} \label{eq:rewardP1} \\
    \quad \text{s. t.} \ \
        &\frac{w^2 r_{ \sft l, 0} }{2 r_0^2}  {\bi}^H \boldm \boldm^H {\bi} \geq \beta_0 \label{eq:const1P1} \\
        &\bi^H {\bb}_n {\bb}_n^H \bi \leq V_{ n}^2, \;\; n=1, \; 2, \; \cdots \; N \label{eq:const2P1} \\
        &\bi^H \bQ_n \bi \leq A_{ n}^2, \;\; n=1, \; 2, \; \cdots \; N \label{eq:const3P1}
\end{align}
\end{subequations}
$\mathrm{(P1)}$ is a complex-valued non-convex QCQP problem \cite{ConvecOptBoyd04}. Although solving such a problem is nontrivial in general \cite{LuoZhangSPM10, HuangPalomar10}, we obtain the optimal solution to $\mathrm{(P1)}$ in Section~\ref{sec: solution}.

\section{Optimal Solutions}\label{sec: solution}
\subsection{Optimal Solution to ($\mathrm{P1}$) without  Constraints~\eqref{eq:const2P1} and~\eqref{eq:const3P1}}
In this subsection, we consider the simplified ($\mathrm{P1}$) only with delivered load power constraint~\eqref{eq:const1P1} but without voltage and current constraints given in \eqref{eq:const2P1} and \eqref{eq:const3P1}, respectively, to get useful insights on magnetic beamforming. We observe that the real-part currents $\bar{\bi}$ and the imaginary-part currents $\hat{\bi}$ contribute in the same way to the total TX power in \eqref{eq:rewardP1} and the delivered load power in \eqref{eq:const1P1}, since both $\overline{\bB}$ and $\boldm \boldm^H$ are symmetric matrices. As a result, we can set $\hat{\bi}=\mathbf{0}$ and design $\bar{\bi}$ only, i.e., we need to solve
\begin{subequations}\label{eq:optimP2}
\begin{align}
  \mathrm{(P2)}: \ \ \ \underset{ \bar{\bi}  \in \bbR^N}{\text{min}} \ \
  & \frac{1}{2}  {\bar{\bi}}^H  {\overline{\bB}} {\bar{\bi}}   \label{eq:rewardP2} \\
   \quad \text{s. t.} \ \
       &\frac{w^2 r_{ \sft l, 0} }{2 r_0^2} \bar{\bi}^H \boldm \boldm^H \bar{\bi} \geq \beta_0 \label{eq:const1P2}
\end{align}
\end{subequations}
The optimal solution to $\mathrm{(P2)}$ is given as follows.
\begin{mythe}\label{the:solution_specialP2}
The optimal solution to $\mathrm{(P2)}$ is ${\bar{\bi}}^{\star} = \alpha \bu_1$, where $\alpha$ is a constant such that the constraint \eqref{eq:const1P2} holds with equality, and $\bu_1$ is the eigenvector associated with the minimum eigenvalue, denoted by $\gamma_1$, of the matrix
\begin{align}
  \bT=\bR + \frac{w^2 (r_0 -v^{\star} r_{ \sft l, 0})}{r_0^2} \boldm \boldm^H, \label{eq:solution_wo_constraint}
\end{align}
where $v^{\star}$ is chosen such that $\gamma_1 =0$.

Specifically, for the case of identical TX resistances,  i.e., $\bR=r \bI$,
the optimal currents to ($\mathrm{P2}$) is
 \begin{align}
   {\bar{\bi}}^{\star} = \frac{\alpha \boldm}{\| \boldm \|_2}\label{eq:solution_wo_constraint_common_r}.
 \end{align}
\end{mythe}


\begin{proof}
We construct the Lagrange function of $\mathrm{(P2)}$ as 
\begin{align}
  L(\bar{\bi}, v) &= \frac{1}{2} {\bar{\bi}}^H  {\overline{\bB}} {\bar{\bi}} + v \left( \beta_0 - \frac{w^2 r_{ \sft l, 0} }{2r_0^2} \bar{\bi}^H \boldm \boldm^H \bar{\bi}\right) \label{eq:Lagrangian} 
\end{align}

Then, the Lagrange dual function is given by
\begin{align}
L(v) &= \beta_0 v + \underset{{\bar{\bi}}}{\inf} \;\left( \frac{1}{2} {\bar{\bi}}^H  {\overline{\bB}} {\bar{\bi}} -  \frac{w^2 r_{ \sft l, 0} v }{2r_0^2} \bar{\bi}^H \boldm \boldm^H \bar{\bi} \right) \nonumber \\
  &= \beta_0 v \!+\! \underset{{\bar{\bi}}}{\inf} \;\left( \frac{1}{2} {\bar{\bi}}^H  \left(\bR + \frac{w^2 \boldm \boldm^H}{r_{ \sft l, 0} + r_{ \sft p, 0}}\right) {\bar{\bi}} \!-\!  \frac{w^2 r_{ \sft l, 0} v }{2r_0^2} \bar{\bi}^H \boldm \boldm^H \bar{\bi} \right) \nonumber \\
    &=  \beta_0 v \!+\! \underset{{\bar{\bi}}}{\inf} \;\frac{1}{2} {\bar{\bi}}^H  \! \left( \bR \!+\! \frac{w^2 ((1 \!-\! v) r_{ \sft l, 0} \!+\! r_{ \sft p, 0})}{r_0^2} \boldm \boldm^H \right) \bar{\bi}. \label{eq:LagrangeDual_1}
\end{align}
To obtain the best lower bound on the optimal objective value of ($\mathrm{P2}$), the dual variable $v$ should be optimized to maximize the Lagrange dual~\eqref{eq:LagrangeDual_1}. For dual feasibility, the Lagrange dual function~\eqref{eq:LagrangeDual_1} should be bounded below. For convenience, we denote the following singular-value-decomposition (SVD)
\begin{align}
 \bR + \frac{w^2 ((1-v) r_{ \sft l, 0} + r_{ \sft p, 0})}{r_0^2} \boldm \boldm^H = \bU \bGamma \bU^H, \label{eq:SVD_specialP3}
\end{align}
where the matrix $\bU = [\bu_1 \; \bu_2 \; \cdots \; \bu_N]$ is orthogonal, and $\bGamma = \diag\{\gamma_1, \cdots, \gamma_N\}$, with $\gamma_1 \leq \gamma_2 \leq \cdots \leq \gamma_N$. For the case of arbitrary transmitter resistances, the Lagrangian in \eqref{eq:Lagrangian} is bounded below in $\bar{\bi}$ and the Lagrange dual function \eqref{eq:LagrangeDual_1} is maximized, only when $v$ is chosen as $v^{\star}$ such that $\gamma_1 =0$.

Moreover, we observe that the objective \eqref{eq:rewardP2} is minimized when the constraint \eqref{eq:const1P2} holds with equality, since both $\overline{\bB}$ and the matrix $\boldm \boldm^H$ are PSD. Hence, the optimal current can be written as ${\bar{\bi}}^{\star} = \alpha \bu_1$, where $\bu_1$ is the eigenvector associated with the eigenvalue $\gamma_1 =0$, and $\alpha$ is a constant such that the constraint \eqref{eq:const1P2} holds with equality.

For the case of identical transmitter resistance, i.e., $\bR = r \bI_N$,
from the isometric property of the identity matrix $\bI_N$, the diagonal matrix $\bGamma$ is given by
\begin{align}
\bGamma &= \diag \left\{ r + \frac{w^2 ((1-v) r_{ \sft l, 0} + r_{ \sft p, 0})}{r_0^2}, \; r,\; \cdots, \; r\right\}
\nonumber
\end{align}
and the eigenvector $\bu_1 = \frac{\boldm}{\| \boldm \|_2}$, and $\bu_n, \; \forall n \geq 2$, are arbitrarily orthogonal vectors constructed by methods such as Gram$-$Schmidt method. {It is standard to show that the Lagrangian in \eqref{eq:Lagrangian} is bounded below in $\bar{\bi}$ and the Lagrange dual function \eqref{eq:LagrangeDual_1} is maximized, only when $v$ is chosen such that the first eigenvalue is zero}, i.e., the optimal dual variable is
\begin{align}
v^{\star} = 1 + \frac{r r_0^2 }{r_{ \sft l, 0} w^2} + \frac{r_{ \sft p, 0}}{r_{ \sft l, 0}},
\end{align}
and the optimal current is given by
\begin{align}
{\bar{\bi}}^{\star} = \frac{\alpha \boldm}{\| \boldm \|_2}, \label{eq:opt_current_wpecial}
\end{align}
where $\alpha$ is a constant such that \eqref{eq:const1P2} holds with equality.
\end{proof}

In fact, all TX currents can be rotated by arbitrarily common phase.
\begin{myrem}
 Theorem \ref{the:solution_specialP2} implies that for the case of identical TX resistances, the optimal current magnitude of each TX $n$ is proportional to the mutual inductance $M_{0n}$ between the RX and TX $n$. This is analogous to the traditional maximum-ratio-transmission (MRT) beamforming in wireless communications~\cite{GershmanSPM10}. The magnetic beamforming also differs from the traditional beamforming in wireless communications. The former operates over the near-field magnetic flux, and only the TX current magnitudes are adjusted according to the real magnetic channels (i.e., positive or negative mutual inductances)~\cite{JadidianKatabi14}; while the latter operates over the far-field propagating waves, and the amplitude as well as the phase of the signal at each TX are adjusted according to complex wireless channels including amplitude fluctuation and phase shift due to prorogation delay~\cite{GoldsmithWC2005}.
\end{myrem}


\subsection{Optimal Solution to ($\mathrm{P1}$)}
Define $\bX \triangleq \bi \bi^H$, $\bM \triangleq \boldm \boldm^H$, and $\bB_n \triangleq \bb_n \bb_n^H$. Thus, $\mathrm{(P1)}$ can be equivalently rewritten as the following SDP problem
\begin{subequations}\label{eq:optimP2-SDP}
\begin{align}
   (\mathrm{P1}\mathrm{-SDP}): \ \ \ &\underset{ \bX \in \bbC^{N \times N}}{\text{min}} \ \
    \frac{1}{2}   \Tr \left( {\overline{\bB}}  \bX \right) \label{eq:rewardP1SDP} \\
    \quad \text{s. t.} \ \
        &\Tr \left( {\bM} \bX \right) \geq \frac{2 r_0^2 \beta_0}{w^2 r_{ \sft l, 0} } \label{eq:const1P1SDP} \\
        &\Tr \left( {\bB_n} \bX \right) \leq V_{ n}^2, \;\; n=1, \; \cdots \; N \label{eq:const2P1SDP}
        \\
        &\Tr \left( \bQ_n \bX \right) \leq A_{ n}^2, \;\; n=1, \; \cdots \; N \label{eq:const3P1SDP} \\
        &  \bX \succcurlyeq 0, \; \rank \left(\bX\right)=1 \label{eq:const4P1SDP}
\end{align}
\end{subequations}
where $\bX \succcurlyeq 0$ indicates that $\bX$ is PSD.



As well known, the rank constraint in \eqref{eq:const4P1SDP} is non-convex \cite{LuoZhangSPM10,HuangPalomar10}. By ignoring the rank-one constraint in~\eqref{eq:const4P1SDP}, we obtain the SDR of $(\mathrm{P1}\mathrm{-SDP})$, denoted by $(\mathrm{P1}\mathrm{-SDR})$, which is convex. Moreover, we have the following theorem on $(\mathrm{P1}\mathrm{-SDR})$ and the optimal solution to $(\mathrm{P1})$.
\begin{mythe}\label{the:rank-one}
  The SDR of $(\mathrm{P1}\mathrm{-SDP})$ is tight, i.e., the existing optimal solution $\bX^{\star}$ to $(\mathrm{P1}\mathrm{-SDR})$ is always rank-one (i.e., $\bX^{\star} = \bi^{\star} \left(\bi^{\star}\right)^H$). The optimal solution to $(\mathrm{P1})$ is $\bi^{\star}$.
\end{mythe}

\begin{proof}
Let $\lambda \geq 0, \ \bm{\rho}=(\rho_1, \; \cdots, \rho_N)^T, \ \forall \rho_n \geq 0, \  \bm{\mu} =(\mu_1, \cdots, \mu_N), \ \forall \mu_n \geq 0,$ be the dual variables corresponding to the constraint(s) given in~\eqref{eq:const1P1},~\eqref{eq:const2P1}, and~\eqref{eq:const3P1}, respectively. Let the matrix $\bS  \succcurlyeq 0$ be the dual variable corresponding to the constraint $\bX \succcurlyeq 0$ in~\eqref{eq:const4P1SDP}. The Lagrangian of $(\mathrm{P1}\mathrm{-SDR})$ is then written as
\begin{align}
  L(\bX, \lambda, \bm{\rho}, \bS)  &= \frac{1}{2}   \Tr \left( {\overline{\bB}}  \bX \right) - \lambda \left( \Tr \left( {\bM} \bX \right) - \frac{2 r_0^2 \beta_0}{w^2 r_{ \sft l, 0} }\right)  \!+ \nonumber \\
  &\quad  \!\! \sum \limits_{n=1}^N \! \rho_n \! \left( \Tr \left( {\bB_n} \bX \right) \!-\!A_{ n}^2 \right) \!+\! \! \sum \limits_{n=1}^N \! \mu_n \!  \left( \Tr \left( {\bQ_n} \bX \right) \!-\! D_{ n}^2 \right) \!-\! \Tr \left( \bS \bX \right).
\end{align}

Let $\bX^{\star}, \lambda^{\star}, \bm{\rho}^{\star}, {\bm{\mu}^{\star}}$, and $\bS^{\star}$ be the optimal primal variables and dual variables, respectively. Moreover, the Karush$-$Kuhn$-$Tucker (KKT) conditions are given by
\begin{align}
  \nabla_{\bX}  L(\bX^{\star}, \lambda^{\star}, \bm{\rho}^{\star}, \bm{\mu}^{\star}, \bS^{\star}) &= \frac{1}{2} {\overline{\bB}} -\lambda^{\star} \bM +  \sum \limits_{n=1}^N \rho_n^{\star} {\bB_n} + \sum \limits_{n=1}^N \mu_n^{\star} {\bQ_n}   - \bS^{\star} = \bm{0}. \label{eq:KKT_orthgonality} \\
   \bS^{\star} \bX^{\star} &= \bm{0}. \label{eq:KKT_complementarity}
\end{align}

Next, by postmultiplexing \eqref{eq:KKT_orthgonality} by $\bX^{\star}$ and substituting \eqref{eq:KKT_complementarity} into the resulting equation, we obtain
\begin{align}
  \frac{1}{2} {\overline{\bB}} \bX^{\star} -\lambda^{\star} \bM \bX^{\star} +  \sum \limits_{n=1}^N \rho_n^{\star} {\bB_n} \bX^{\star} +  \sum \limits_{n=1}^N \mu_n^{\star} {\bQ_n} \bX^{\star}= \bm{0}. \label{eq:KKT_orthgonality2}
\end{align}

\noindent We further have
\begin{align}
\rank \left( \left( \frac{1}{2} {\overline{\bB}}+  \sum \limits_{n=1}^N \rho_n^{\star} {\bB_n}  +  \sum \limits_{n=1}^N \mu_n^{\star} {\bQ_n} \right) \bX^{\star} \right)  =\rank \left( \bM \bX^{\star}\right) \leq \rank \left( \bM \right) =1. \label{eq:KKT_orthgonality3}
\end{align}
Since $\overline{\bB}$ is PSD, the matrix $\left(\! \frac{1}{2} {\overline{\bB}} \!+\!  \sum \limits_{n=1}^N \rho_n^{\star} {\bB_n} \!+\!  \sum \limits_{n=1}^N \mu_n^{\star} {\bQ_n} \! \right)$ has full rank. Hence, the relation \eqref{eq:KKT_orthgonality3} implies that
\begin{align}
 \! \! \! \rank \left( \bX^{\star} \right)  \!= \! \rank \!  \left( \! \left(\! \frac{1}{2} {\overline{\bB}} \!+\!  \sum \limits_{n=1}^N \rho_n^{\star} {\bB_n} \!+\!  \sum \limits_{n=1}^N \mu_n^{\star} {\bQ_n} \!\right) \bX^{\star} \! \right) \!=\! 1. \label{eq:KKT_orthgonality4}
\end{align}
Since $\bX^{\star}$ is rank-one, it can be written as $\bX^{\star} = \bi^{\star} \left(\bi^{\star}\right)^H$. The vector $\bi^{\star}$ is thus the solution to $(\mathrm{P1})$. This completes the proof.
\end{proof}



\begin{myrem}
As a convex problem, $(\mathrm{P1}\mathrm{-SDR})$ can be polynomially solved via an interior-point method~\cite{ConvecOptBoyd04}, to arbitrary accuracy. The optimal solution to $(\mathrm{P1})$ is directly obtained from the optimal solution to $(\mathrm{P1}\mathrm{-SDR})$, without any postprocessing required.
\end{myrem}

\section{Numerical Results}\label{sec: simulation}
As shown in Fig.~\ref{fig:Fig2}, we consider the setup with $N=5$ TX coils and one RX coil, each of which has 100 turns and a radius of 0.1 meter. We use cooper wire with radius of 0.1 millimeter for all coils. All the coils are horizontal with respect to the $xy$-plane. The total resistance of each TX is the same as $0.336 \Omega$. For the RX, its parasitic resistance and load resistance are $r_{ \sft {p}, 0} =0.336 \Omega$ and $r_{\sft{l},0}=50 \Omega$, respectively. The self and mutual inductances are given in Table~\ref{table_inductance}. All the capacitors are chosen such that the resonance angular frequency is $w=6.78 \times 2 \pi$ rad/second \cite{RezaZhang:C15}. We assume that the constraint thresholds $V_n=30\sqrt{2}$~V and $A_{n}=5\sqrt{2}$~A.
\begin{figure} [htbp]
\centering
\includegraphics[width=.65\columnwidth]{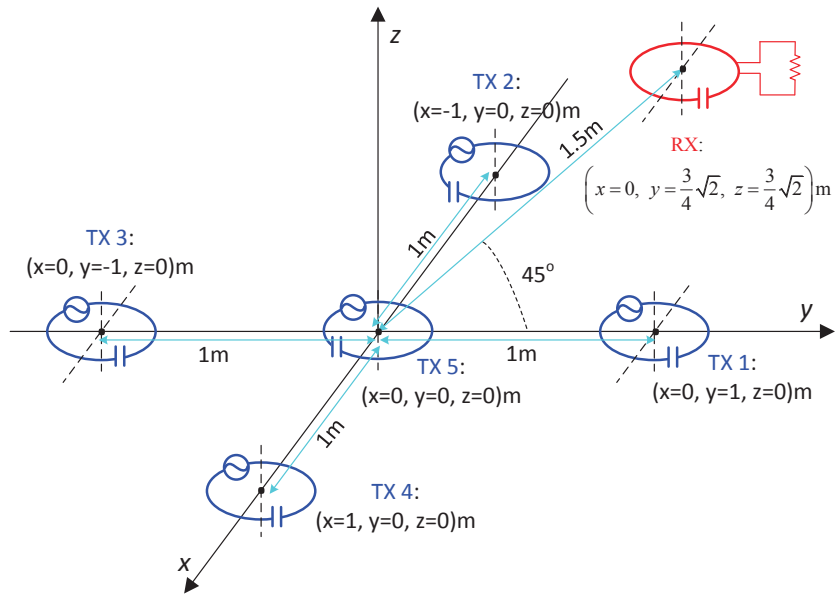}
\caption{System setup}
\label{fig:Fig2}
\end{figure}
\begin{table} [htbp]
\centering
\caption{Mutual/Self inductances ($\mu$H)} \label{table_inductance}
\small{
\begin{tabular}{*{6}{c}}
  \hline \hline
   & TX 1 & TX 2 & TX 3 & TX 4 & TX 5\\
  \hline
  TX 1   & 5886.8  & 0.3565 & 0.1253 & 0.3565  & 0.2984  \\
  TX 2  & 0.3565  & 5886.8 & 0.3565  & 0.1253  & 0.2984  \\
  TX 3   & 0.1253  & 0.3565 & 5886.8  & 0.3565  &  0.2984 \\
  TX 4   & 0.3565  & 0.1253 & 0.3565  & 5886.8  &  0.2984 \\
  TX 5   & 0.2984  & 0.2984 & 0.2984 & 0.2984  & 5886.8  \\
  RX   & 1.6121  & 0.00781 & -0.0296 & 0.00781  & 0.1508  \\
  \hline
\end{tabular}
}
\end{table}


For comparison, we consider an uncoordinated benchmark scheme of equal current allocation, i.e., all TXs carry the same in-phase current which can be adjusted. In particular, we compare the WPT performance for the case with TXs' constraints on the peak voltage and current, and for the case without TXs' constraints, respectively. We define the efficiency of WPT as the ratio of the delivered load power $\beta_0$ to the total TX power $p$, i.e., $\eta \triangleq \frac{\beta_0}{p}$.


Fig.~\ref{fig:Fig_A_fixedR} plots the total TX power $p$ and the efficiency $\eta$ versus the delivered load power $\beta_0$. All curves show the feasible and optimal values. For the case without TXs' constraints, we observe that the WPT efficiencies by using magnetic beamforming and by using the benchmark are 87.1$\%$ and 62.0$\%$, respectively. Hence, the exact gain of magmatic beamforming is an increase of WPT efficiency by 25.1$\%$.

For the case with TXs' constraints, we observe that the magnetic beamforming can deliver more power up to 73.6~W to the RX with the efficiency of 74.1$\%$; while the benchmark can deliver at most 36.0~W to the RX load with the efficiency of 62.0$\%$. That is, the deliverable power is enhanced by 104.4$\%$ by using magnetic beamforming, over the uncoordinated benchmark of equal current allocation. Thus, besides the efficiency improvement, another important benefit of magnetic beamforming is the enhancement of the deliverable power, under practical TX constraints.


\begin{figure} [htbp]
\centering
\includegraphics[width=.65\columnwidth]{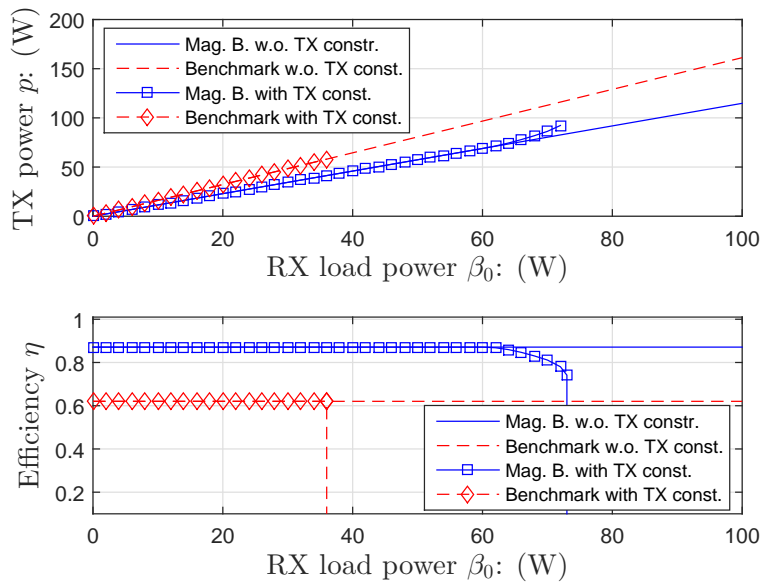}
\caption{Power and efficiency v.s. RX power $\beta_0$.} 
\label{fig:Fig_A_fixedR}
\end{figure}

Fig.~\ref{fig:Fig_A_fixedR} also shows that the WPT efficiency decreases for $60.5 < \beta_0  < 73.6$. To explain the decreasing efficiency and obtain insights for magnetic beamforming, we investigate the cases of $\beta_0=60$~W and 73.5~W in the following. The optimal currents and power of all TXs are obtained in Table~\ref{table_solution}. For $\beta_0=60$~W, TX 1 carries almost the peak current, and most energy is consumed by TX 1 that has the largest mutual inductance with the RX. This implies that the TXs with larger mutual inductances with the RX are favorable to carry higher currents, to achieve high efficiency of WPT. For $\beta_0=60$~W, all TX constraints are inactive, and it can be further checked that the current of each TX is proportional to its mutual inductance with the RX. This verifies \eqref{eq:solution_wo_constraint_common_r} in Theorem \ref{the:solution_specialP2}. To support higher RX load power of 73.5~W, we observe that TX 1, 3 and 5 carry the (maximally allowable) peak current, and TX 2 as well as TX 4 increase their carried currents. The cost is the decreased overall efficiency, due to smaller mutual inductance between TX 2, 3, 4, 5 and the RX, than that between TX 1 and the RX.

\begin{table} [htbp]
\centering
\caption{Optimal results for $\beta_0=50, 56$.} \label{table_solution}
\small{
\begin{tabular}{*{3}{c}}
  \hline \hline
   & $\beta_0=60$ & $\beta_0=73.5$\\
  \hline
  $(i_1^{\star}, p_1^{\star})$    & $(-7.0698, 68.25)$ & $(-7.071, 74.66)$ \\
  $(i_2^{\star}, p_2^{\star})$   & $(-0.0342, 0.0016)$ & $(-3.499, 2.22)$ \\
  $(i_3^{\star}, p_3^{\star})$   & $(0.1296, 0.023)$ & $(-7.071, 9.62)$ \\
  $(i_4^{\star}, p_4^{\star})$   & $(-0.0342, 0.0016)$ & $(-3.499, 2.22)$\\
    $(i_5^{\star}, p_5^{\star})$   & $(-0.6617,  0.598)$ & $(-7.071, 14.60)$\\
  \hline
\end{tabular}
}
\end{table}

\section{Conclusion}
This paper has studied the optimal magnetic beamforming for an MRC-WPT system with multiple TXs and a single RX. We formulate an optimization problem to minimize the total power drawn from the voltage sources of all TXs by jointly designing the currents flowing through different TXs, subject to the RX constraint on the required load power and the TXs' constraints on the peak voltage and current. For the special case of identical TX resistances and neglecting TXs' constraints, the optimal current of each TX is shown to be proportional to the mutual inductance between its TX coil and the RX coil. In general, the formulated non-convex QCQP problem is recast as an SDP problem. Its SDR is shown to be tight. Numerical results show that magnetic beamforming significantly enhances the deliverable power and the WPT efficiency, compared to the uncoordinated benchmark scheme of equal current allocation. Furthermore, it is numerically shown that TXs with large mutual inductances with the RX are favorable to carry high currents, to achieve high efficiency for WPT.

\renewcommand{\baselinestretch}{1.45}
\bibliography{IEEEabrv,reference_150528}
\bibliographystyle{IEEEtran}

\end{document}